\documentclass[a4paper]{article}
\usepackage[pdftex]{graphicx}
\usepackage{epstopdf}
\expandafter\let\csname equation*\endcsname\relax
\expandafter\let\csname endequation*\endcsname\relax
\usepackage{amsmath}
\usepackage{amsfonts}
\usepackage{amssymb}
\usepackage{makeidx}
\usepackage[margin=3cm]{geometry}

\newcommand{\beq}{\begin{equation}} \newcommand{\eeq}{\end{equation}}
\newcommand{\bea}{\begin{eqnarray}} \newcommand{\eea}{\end{eqnarray}}
\newcommand{\bear}{\begin{eqnarray*}} \newcommand{\eear}{\end{eqnarray*}}
\newcommand{\lb}{\label} 
\newcommand{\rf}[1]{(\ref{#1})}   

\newtheorem{theorem}{Theorem}

\newtheorem{definition}{Definition}
\newtheorem{remark}{Remark}

\newenvironment{proof}[1][Proof]{\begin{trivlist}
\item[\hskip \labelsep {\bfseries #1}]}{\end{trivlist}}

\begin{document}

\title {Fractional variational problems depending on fractional derivatives of differentiable functions with application to  nonlinear chaotic systems}

\author{Matheus Jatkoske Lazo\\
  Instituto de Matem\'atica, Estat\'istica e F\'isica,\\
  Universidade Federal do Rio Grande,\\
  Rio Grande - RS,\\
  Brazil.\\
  \texttt{matheuslazo@furg.br}}

\date{}

\maketitle

\begin{abstract}

\indent
In the present work, we formulate a necessary condition for functionals with Lagrangians depending on fractional derivatives of differentiable functions to possess an extremum. The Euler-Lagrange equation we obtained generalizes previously known results in the literature and enables us to construct simple Lagrangians for nonlinear systems.  As examples of application, we obtain Lagrangians for some chaotic dynamical systems.


\end{abstract}




\section{Introduction}

The calculus with fractional derivatives and integrals of non-integer order started more than three centuries ago when Leibniz proposed a derivative of order $\frac{1}{2}$ in response to a letter from l'H\^opital \cite{OldhamSpanier}. This subject was also considered by several mathematicians as Euler, Fourier, Liouville, Grunwald, Letnikov, Riemann and others up to nowadays. Although the fractional calculus is almost as old as the usual integer order calculus, only in the last three decades it has gained more attention due to its applications in various fields of science (see \cite{SATM,Kilbas,Diethelm,Hilfer,Magin,SKM} for a review). Fractional derivatives are generally non-local operators and are historically applied to study non-local or time dependent processes, as well as to model phenomena involving coarse-grained and fractal spaces. As an example, applications of fractional calculus in coarse-grained and fractal spaces are found in the framework of anomalous diffusion \cite{Metzler,Metzler2,Klages} and field theories \cite{Tarasov3,Lazo,BI,Calcagni1,Calcagni2,Vacaru}. 

The fractional calculus of variation was introduced in the context of classical mechanics. Riewe \cite{Riewe1,Riewe2} showed that a Lagrangian involving fractional time derivatives leads to an equation of motion with non-conservative forces such as friction. It is a remarkable result since frictional and non-conservative forces are beyond the usual macroscopic variational treatment \cite{Bauer}, and consequently, beyond the most advanced methods of classical mechanics. Riewe generalized the usual calculus of variations for a Lagrangian depending on fractional derivatives \cite{Riewe1,Riewe2} in order to deal with linear non-conservative forces. Recently, several approaches have been developed to generalize the least action principle and the Euler-Lagrange equations to include fractional derivatives \cite{Agrawal,BA,Cresson,APT,AT,OMT,MalinowskaTorres}. 

Despite the Riewe approach has been successfully applied to study open and/or non-conservative linear systems, it cannot be directly applied to nonlinear open systems. The limitation follows from the fact that, in order to obtain a final equation of motion containing only integer order derivatives, the Lagrangian should contain only quadratic terms depending on fractional derivatives. In the present work we formulated a generalization of Riewe fractional action principle by taking advantage of a so called practical limitation of fractional derivatives, namely, the absence of a simple chain and Leibniz's rules.

As examples, we applied our generalized fractional variational principle to some nonlinear chaotic third-order dynamical systems, so called jerk dynamical systems because the derivative of the acceleration with respect to time is referred to as the jerk \cite{jerk}. These systems are important because they are the simplest ever one-dimensional autonomous ordinary differential equations which display dynamical behaviors including chaotic solutions \cite{jerk1,jerk2a,jerk2b,jerk2c,jerk2d,jerk2e,jerk2f,jerk2g}. It is important to mention that jerk dynamical systems describe several phenomena in physics, engineering and biology, such as electrical circuits, mechanical oscillators, laser physics, biological systems, etc \cite{jerk1,jerk2a,jerk2b,jerk2c,jerk2d,jerk2e,jerk2f,jerk2g}.


\section{The Riemann-Liouville and Caputo Fractional Calculus}

The fractional calculus of derivative and integration of non-integers orders started more than three centuries ago with l'H\^opital and Leibniz when the derivative of order $\frac{1}{2}$ was suggested \cite{OldhamSpanier}. This subject as also considered by several mathematicians as Euler, Laplace, Liouville, Grunwald, Letnikov, Riemann and others up to nowadays. Although the fractional calculus is almost as old as the usual integer order calculus, only in the last three decades it has gained more attention due to its applications in various fields of science, engineering, economics, biomechanics, etc (see \cite{SATM,Kilbas,Hilfer,Magin} for a review). Actually, there are several definitions of fractional order derivatives. These definitions include the Riemann-Liouville, Caputo, Riesz, Weyl,  Grunwald-Letnikov, etc. (see \cite{OldhamSpanier,SATM,Kilbas,Diethelm,Hilfer,Magin,SKM} for a review). In this section we review some definitions and properties of the Riemann-Liouville and Caputo fractional calculus.

Despite we have many different approaches to fractional calculus, several known formulations are somehow connected with the analytical continuation of Cauchy formula for $n$-fold integration
\beq
\lb{a2}
\begin{split}
\int_{a}^t x(\tilde{t})(d\tilde{t})^{n} &= \int_{a}^t\int_{a}^{t_{n}}\int_{a}^{t_{n-1}}\cdots \int_{a}^{t_3}\int_{a}^{t_2} x(t_1)dt_1dt_2\cdots dt_{n-1}dt_{n} \\
&= \frac{1}{\Gamma(n)}\int_{a}^t \frac{x(u)}{(t-u)^{1-n}}du \;\;\;\;\; (n\in \mathbb{N}),
\end{split}
\eeq
where $\Gamma$ is the Euler gamma function. The proof of Cauchy formula can be found in several textbooks (for example, it can be found in \cite{OldhamSpanier}). The analytical continuation of \rf{a2} give us a definition for an integration of non-integer (or fractional) order. This fractional order integration is the building bloc of the Riemann-Liouville and Caputo calculus, the two most popular formulations of fractional calculus, as well as several other approaches to fractional calculus \cite{OldhamSpanier,SATM,Kilbas,Diethelm,Hilfer,Magin,SKM}. The fractional integration obtained from \rf{a2} is historically called Riemann-Liouville left and right fractional integrals: 
\begin{definition} Let $\alpha \in \mathbb{R}_+$. The operators ${_a J^{\alpha}_t}$ and ${_t J^{\alpha}_b}$ defined on $L_1[a,b]$ by
\beq
\lb{a3}
{_a J^{\alpha}_t} x(t) =\frac{1}{\Gamma(\alpha)}\int_{a}^t \frac{x(u)}{(t-u)^{1-\alpha}}du \;\;\;\;\; (\alpha \in \mathbb{R}_+)
\eeq
and
\beq
\lb{a4}
{_t J^{\alpha}_b} x(t) =\frac{1}{\Gamma(\alpha)}\int_t^b \frac{x(u)}{(u-t)^{1-\alpha}}du \;\;\;\;\; ( \alpha \in \mathbb{R}_+),
\eeq
with $a<b$ and $a,b\in \mathbb{R}$, are called left and the right fractional Riemann-Liouville integrals of order $\alpha \in \mathbb{R}$, respectively.
\end{definition}

For integer $\alpha$ the fractional Riemann-Liouville integrals \rf{a3} and \rf{a4} coincide with the usual integer order $n$-fold integration \rf{a2}. Moreover, from the definitions \rf{a3} and \rf{a4} it is easy to see that the Riemann-Liouville fractional integrals converge for any integrable function $x$ if $\alpha>1$. Furthermore, it is possible to proof the convergence of \rf{a3} and \rf{a4} for $x\in L_1[a,b]$ even when $0<\alpha<1$ \cite{Diethelm}.

The integration operators ${_a J^{\alpha}_t}$ and ${_t J^{\alpha}_b}$ play a fundamental role in the definition of fractional Riemann-Liouville and Caputo calculus. In order to define the Riemann-Liouville derivatives, we recall that for positive integers $n>m$ it follows the identity $D^m_t x(t)=D^{n}_t {_aJ^{n-m}_t x(t)}$, where $D^m_t$ is an ordinary derivative of integer order $m$. 
\begin{definition}[Riemann-Liouville]
The left and the right Riemann-Liouville fractional derivative of order $\alpha >0$ ($\alpha\in \mathbb{R}$) are defined, respectively, by ${_a D^{\alpha}_t} x(t) := D^{n}_t {_a J^{n-\alpha}_t} x(t)$ and ${_t D^{\alpha}_b} x(t):=(-1)^nD^{n}_t{_t J^{n-\alpha}_b} x(t)$ with $n=[\alpha]+1$, namely
\beq
\lb{a5}
{_a D^{\alpha}_t} x(t)=\frac{1}{\Gamma(n-\alpha)}\frac{d^n}{dt^n}\int_{a}^t \frac{x(u)}{(t-u)^{1+\alpha-n}}du 
\eeq
and
\beq
\lb{a6}
{_t D^{\alpha}_b} x(t)=\frac{(-1)^n}{\Gamma(n-\alpha)}\frac{d^n}{dt^n}\int_{t}^b \frac{x(u)}{(u-t)^{1+\alpha-n}}du,
\eeq
where $\frac{d^n}{dt^n}$ stands for ordinary derivatives of integer order $n$. 
\end{definition}
On the other hand, the Caputo fractional derivatives are defined by inverting the order between derivatives and integrations.
\begin{definition}[Caputo]
The left and the right Caputo fractional derivatives of order $\alpha\in \mathbb{R}_+$ are defined, respectively,
by ${_a^C D}^{\alpha}_t x(t) := {_aJ}^{n-\alpha}_t D^{n}_t x(t)$
and ${_t^C}D^{\alpha}_b x(t) := (-1)^n _tJ^{n-\alpha}_b
D^{n}_t x(t)$ with $n=[\alpha]+1$; that is,
\begin{equation}
\label{a7}
{_a^C D}^{\alpha}_t x(t) := \frac{1}{\Gamma(n-\alpha)}
\int_{a}^t \frac{x^{(n)}(u)}{(t-u)^{1+\alpha-n}}du
\end{equation}
and
\begin{equation}
\label{a8}
{_t^C D}^{\alpha}_b x(t)
:= \frac{(-1)^n}{\Gamma(n-\alpha)}\int_{t}^b
\frac{x^{(n)}(u)}{(u-t)^{1+\alpha-n}}du,
\end{equation}
where $a \le t \le b$ and $x^{(n)}(u)=\frac{d^n x(u)}{du^n} \in L_1([a,b])$
is the ordinary derivative of integer order $n$.
\end{definition}

An important consequence of definitions \eqref{a5}--\eqref{a8} is that the Riemann-Liouville and Caputo fractional derivatives are non-local operators. The left (right) differ-integration operator \eqref{a5} and \eqref{a7} (\eqref{a6} and \eqref{a8}) depends on the values of the function at left (right) of $t$, i.e. $a\leq u \leq t$ ($t\leq u \leq b$). On the other hand, it is important to note that when $\alpha$ is an integer, the Riemann-Liouville fractional derivatives \eqref{a5} and \eqref{a6} reduce to ordinary derivatives of order $\alpha$. On the other hand, in that case, the Caputo derivatives \eqref{a7} and \eqref{a8} differ from integer order ones by a polynomial of order $\alpha -1$ {\rm \cite{Kilbas,Diethelm}.

It is important to remark, for the purpose of this work, that the fractional derivatives \rf{a5}--\rf{a8} do not satisfy a simple generalization of the chain and Leibniz's rules of classical derivatives \cite{OldhamSpanier,SATM,Kilbas,Diethelm,Hilfer,Magin,SKM}. In other words, generally we have:
\begin{equation}
\label{a9}
{_a^C D}^{\alpha}_t \left[x(t)y(t)\right]\neq y(t){_a^C D}^{\alpha}_tx(t)+x(t){_a^C D}^{\alpha}_ty(t)
\end{equation}
and
\begin{equation}
\label{a10}
{_a^C D}^{\alpha}_t y(x(t)) \neq {_a^C D}^{\alpha}_u y(u)|_{u=x} \;{_a^C D}^{\alpha}_t x(t).
\end{equation}
The absence of a simple chain and Leibniz's rules is commonly considered a practical limitation of the fractional derivatives \rf{a5}--\rf{a8}. However, in the present work we take advantage of this limitation in order to formulate a generalized Lagrangians for nonlinear systems.

In addition to the definitions \rf{a5}--\rf{a8}, we make use of the following property in order to obtain a fractional generalization of the Euler-Lagrange condition.
\begin{theorem}[Integration by parts --- see, e.g., \cite{SKM}]
\label{thm:ml:03}
Let $0<\alpha<1$ and $x$ be a differentiable function in $[a,b]$ with $x(a)=x(b)=0$. For any function $y \in L_1([a,b])$ one has
\begin{equation}
\label{a15}
\int_{a}^{b} y(t) {_a^C D_t^{\alpha}} x(t)dt
= \int_a^b x(t) {_t D_b^{\alpha}} y(t)dt
\end{equation}
and
\begin{equation}
\label{a16}
\int_{a}^{b} y(t)  {_t^C D_b^\alpha} x(t)dt
=\int_a^b x(t) {_a D_t^\alpha} y(t) dt.
\end{equation}
\end{theorem}

It is important to notice that the formulas of integration by parts \eqref{a15} and \eqref{a16}
relate Caputo left (right) derivatives to Riemann-Liouville right (left) derivatives.

Finally, in order to obtain the equation of motion for our examples, we are going to use the following two relations 
\beq
\lb{a17}
i {_a^C D_t^{\frac{1}{2}}} x(t) = {_t^C D_b^{\frac{1}{2}}}x(t) \quad \quad \mbox{in the limit} \quad \quad a\rightarrow b
\eeq
and
\beq
\lb{a18}
{_t D_b^{\frac{1}{2}}} {_t^C D_b^{\frac{1}{2}}} x(t)=\frac{d}{dt}{_t J^{\frac{1}{2}}_b}{_t J^{\frac{1}{2}}_b}\frac{d}{dt}x(t)=\frac{d}{dt}{_t J^{1}_b}\frac{d}{dt}x(t)=\frac{d}{dt}x(t)
\eeq
The proof of \rf{a17} can be found in \cite{TRG}, and \rf{a18} follows from the general semi-group property ${_t J^{\alpha}_b}{_t J^{\beta}_b}={_t J^{\alpha+\beta}_b}$ (see, e.g., \cite{SKM,Diethelm}).


\section{A Generalized Fractional Lagrangian}

In classical calculus of variations it is of no conceptual and practical importance to deal with Lagrangian functions depending on derivatives of nonlinear functions of the unknown function $x$. This is due to the fact that in these cases we can always rewrite the Lagrangian $L$ as an usual Lagrangian $\tilde{L}$ by applying the chain's rules. As for example, for a differentiable function $f$ we can rewrote $L(t,x,\frac{d}{dt}f(x))=L(t,x,\frac{d}{du}f(u)|_{u=x}\dot{x})=\tilde{L}(t,x,\dot{x})$, where $\frac{dx}{dt}=\dot{x}$. However, this simplification for the fractional calculus of variation is not possible due to the absence of a simple chain's rule for fractional derivatives. It is just this apparent limitation of fractional derivatives what opens the very interesting possibility to investigate new kinds of Lagrangian suitable to study nonlinear systems. In the present work we investigate for the first time these kinds of Lagrangian and we apply them to construct Lagrangians for some Jerk nonlinear dynamical system. 

Our main result is the following Theorem:
\begin{theorem}
Let $f,g:\mathbb{R}\rightarrow \mathbb{R}$ be differentiable functions, and $S$ an action of the form
\begin{equation}
\label{t1}
S=\int_a^b {L} \left(t,x,\dot{x},{_a^C D_t^{\alpha}} f(x),{_a^C D_t^{\alpha}} g(\dot{x})\right)dt,
\end{equation}
where ${_a^C D_t^{\alpha}}$ is a Caputo fractional derivative of order $0<\alpha<1$, and the function $x$ satisfy the fixed boundary conditions $x(a)=x_a$, $x(b)=x_b$, and $\dot{x}(a)=\dot{x}_a$ $\dot{x}(b)=\dot{x}_b$. Also let ${L}\in C^{2}[a,b]\times \mathbb{R}^{4}$. Then, the necessary condition for $S$ to possess an extremum at $x$ is that the function $x$ fulfills the following fractional Euler-Lagrange equation:
\begin{equation}
\label{t2}
\frac{\partial {L}}{\partial x}-\frac{d}{dt}\frac{\partial {L}}{\partial \dot{x}}+\frac{d f}{d x}{_t D^{\alpha}_b}\frac{\partial {L}}{\partial\left({_a^C D^{\alpha}_t} f\right)}-\frac{d}{dt}\left(\frac{d g}{d \dot{x}}{_t D^{\alpha}_b}\frac{\partial {L}}{\partial\left({_a^C D^{\alpha}_t} g\right)}\right)=0.
\end{equation}
\end{theorem}
\begin{proof}
In order to develop the necessary conditions for the extremum of the action \rf{t1}, we define a family of functions $x$ (weak variations)
\beq
\lb{p1}
x=x^*+\varepsilon \eta,
\eeq
where $x^*$ is the desired real function that satisfy the extremum of \rf{t1}, $\varepsilon \in \mathbb{R}$ is a constant, and the function $\eta$ defined in $[a,b]$ satisfy the boundary conditions
\beq
\lb{p2}
\eta(a)=\eta(b)=0, \;\;\;\;\;\; \dot{\eta}(a)=\dot{\eta}(b)=0.
\eeq
The condition for the extremum is obtained when the first G\^ateaux variation is zero
\beq
\lb{p3}
\begin{split}
\delta S&=\lim_{\varepsilon \rightarrow 0} \frac{S[x^*+\varepsilon \eta]-S[x^*]}{\varepsilon}=\int_a^b \left[ \eta \frac{\partial L}{\partial x^*} +\dot{\eta}\frac{\partial {L}}{\partial\left(\dot{x}^*\right)} \right.\\
&\left. \;\;\;\;\;\;\;\; +\left({_a^C D_t^{\alpha}}\eta\frac{df}{dx^*}\right)\frac{\partial L}{\partial \left( {_a^C D_t^{\alpha}} f\right)}+\left({_a^C D_t^{\alpha}}\dot{\eta}\frac{dg}{d\dot{x}^*}\right)\frac{\partial L}{\partial \left( {_a^C D_t^{\alpha}} g\right)}\right]dt=0. 
\end{split}
\eeq
Since the function $\eta$ satisfies both $\eta(a)=\eta(b)=0$ and $\dot{\eta}(a)=\dot{\eta}(b)=0$ boundary conditions \rf{p2}, we can use the fractional integration by parts \rf{a15} and \rf{a16} in \rf{p3}, obtaining:
\beq
\lb{p4}
\begin{split}
\delta S&=\int_a^b \left[ \eta \frac{\partial L}{\partial x^*} +\dot{\eta}\frac{\partial {L}}{\partial\left(\dot{x}^*\right)} +\eta\frac{df}{dx^*}{_t D_b^{\alpha}}\frac{\partial L}{\partial \left( {_a^C D_t^{\alpha}} f\right)}+\dot{\eta}\frac{dg}{d\dot{x}^*}{_t D_b^{\alpha}}\frac{\partial L}{\partial \left( {_a^C D_t^{\alpha}} g\right)}\right]dt\\
&=\int_a^b \eta \left[\frac{\partial L}{\partial x^*} -\frac{d}{dt}\frac{\partial {L}}{\partial\left(\dot{x}^*\right)} +\frac{df}{dx^*}{_t D_b^{\alpha}}\frac{\partial L}{\partial \left( {_a^C D_t^{\alpha}} f\right)}-\frac{d}{dt}\left(\frac{dg}{d\dot{x}^*}{_t D_b^{\alpha}}\frac{\partial L}{\partial \left( {_a^C D_t^{\alpha}} g\right)}\right)\right]dt=0,
\end{split}
\eeq
where an additional usual integration by parts was performed in the terms containing $\dot{\eta}$. Finally, by using the fundamental lemma of the calculus of variations, we obtain the fraction Euler-Lagrange equations \rf{t2}. 
\end{proof}

It is important to notice that our Theorem can be easily extended for Lagrangians depending on left Caputo derivatives, and Riemann-Liouville fractional derivatives. Actually, it is also easy to generalize in order to include a nonlinear function $g\left({_a^C D_t^{\alpha}} x\right)$ instead of $g(\dot{x})$. Finally, it is important to mention that our Theorem generalizes \cite{Riewe1,Riewe2} and the more general formulation proposed in \cite{Agrawal}, as well as the Lagrangian formulation for higher order linear open systems \cite{LazoCesar} 
(for a review in recent advances in calculus of variations with fractional derivatives see \cite{MalinowskaTorres}). 

\begin{remark} For $f(x)=g(\dot{x})=0$ our condition \rf{t2} reduces to the ordinary Euler-Lagrange equation, and the boundary conditions $x(a)=x_a$, $x(b)=x_b$ and $\dot{x}(a)=\dot{x}_a$, $\dot{x}(b)=\dot{x}_b$ are defined by only two arbitrary parameter. Note that for this particular case the Euler-Lagrange equation is a second order ordinary differential equation whose solution $x(t)$ is fixed by two parameters. For example, by imposing the conditions $x(a)=x_a$ and $x(b)=x_b$ the solution $x(t)$ is fixed and, consequently, the numbers $\dot{x}_a$ and $\dot{x}_b$ are automatically fixed as functions of $x_a$ and $x_b$.
\end{remark}


\section{Lagrangian For Nonlinear Chaotic Jerk Systems}

As an example for application of our generalized Euler-Lagrange equation \rf{t2}, in this section we obtained Lagrangians for some Jerk systems. The first example is the simplest one-dimensional family of jerk systems that displays chaotic solutions \cite{jerk,jerk1,jerk2a,jerk2b,jerk2c,jerk2d,jerk2e,jerk2f,jerk2g}
\beq
\lb{c5}
\dddot{x}+A\ddot{x}+\dot{x}=G(x), 
\eeq
where $A$ is a system parameter, and $G(x)$  is a nonlinear function containing one nonlinearity, one system parameter and a constant term. A Lagrangian for this jerk system is given by
\begin{equation}
\lb{c6}
L\left(x,{_t^C D^{\frac{1}{2}}_b}x,{_t^C D^{\frac{1}{2}}_b} \dot{x}\right)=-\frac{i}{2}\left({_a^C D^{\frac{1}{2}}_t} \dot{x}\right)^2-\frac{A}{2}\left(\dot{x}\right)^2+\frac{i}{2}\left({_a^C D^{\frac{1}{2}}_t} x\right)^2-\int G(x)dx.
\end{equation}
In order to show that \rf{c6} give us \rf{c5}, we insert \rf{c6} into our generalized Euler-Lagrange equation \rf{t2}, obtaining
\beq
\lb{c7}
i\frac{d}{dt}\left({_t^C D^{\frac{1}{2}}_b} {_a^C D^{\frac{1}{2}}_t}\right) \dot{x}+A\ddot{x}+i\left({_t^C D^{\frac{1}{2}}_b} {_a^C D^{\frac{1}{2}}_t}\right) x=G(x),
\eeq
and we follow the procedure introduced in \cite{Riewe1,Riewe2} by taking the limit $a\rightarrow b$. Taking the limit in \rf{c7} and using \rf{a17} and \rf{a18} we get \rf{c5}.

The Lagrangian \rf{c6} is equivalent to the introduced by us in \cite{LazoCesar}. However, it is important to stress that \rf{c5} is the only chaotic Jerk system containing nonlinearity depending only on $x$ \cite{jerk,jerk1,jerk2a,jerk2b,jerk2c,jerk2d,jerk2e,jerk2f,jerk2g}. For Jerk systems with more complex nonlinearities, as for example $x\dot{x}$, $\dot{x}^2$ and $x\ddot{x}$, it is not possible to formulate a simple Lagrangian, depending only on $x$ and its derivatives, by using classical calculus of variation or previous formulations including fractional derivatives. Using our Euler-Lagrange equation \rf{t2} we can formulate, by the first time, a Lagrangian for these jerk systems \cite{jerk,jerk1,jerk2a,jerk2b,jerk2c,jerk2d,jerk2e,jerk2f,jerk2g}:
\begin{equation}
L=-\frac{i}{2}\left({_a^C D^{\frac{1}{2}}_t} \dot{x}\right)^2-\frac{A}{2}\left(\dot{x}\right)^2+\frac{i}{2}\left({_a^C D^{\frac{1}{2}}_t} x^2\right){_a^C D^{\frac{1}{2}}_t} \dot{x}+\frac{x^2}{2} \Longrightarrow \dddot{x}+A\ddot{x}-\dot{x}^2+x=0,
\end{equation}

\begin{equation}
\begin{split}
L=-\frac{i}{2}\left({_a^C D^{\frac{1}{2}}_t} \dot{x}\right)^2-\frac{A}{2}\left(\dot{x}\right)^2-\frac{i}{2}\left({_a^C D^{\frac{1}{2}}_t} x^2\right){_a^C D^{\frac{1}{2}}_t} x+\frac{x^2}{2} \Longrightarrow \dddot{x}+A\ddot{x}-x\dot{x}+x=0
\end{split}
\end{equation}

\begin{equation}
\begin{split}
L=-\frac{i}{2}\left({_a^C D^{\frac{1}{2}}_t} \dot{x}\right)^2-\frac{A}{2}x\left(\dot{x}\right)^2+i\frac{A+2}{4}\left({_a^C D^{\frac{1}{2}}_t} x^2\right){_a^C D^{\frac{1}{2}}_t} \dot{x}+\frac{x^2}{2} \Longrightarrow \dddot{x}+Ax\ddot{x}-\dot{x}^2+x=0
\end{split}
\end{equation}


\section{Conclusions}

In the present work we obtained an Euler-Lagrange equation for Lagrangians depending on fractional derivatives of nonlinear functions of the unknown function $x$. Our formulation enables us to obtain Lagrangians for nonlinear open and dissipative systems, and consequently, it enables us to use the most advanced methods of classical mechanic to study these systems. As examples of applications, we obtained a Lagrangian for some chaotic Jerk dynamical system.


\section*{Acknowledgments}

The author is grateful to the Brazilian foundations FAPERGS, CNPq and Capes for financial support.



\begin{thebibliography}{}

\bibitem{OldhamSpanier} K. B. Oldham and J. Spanier, {\it The Fractional Calculus}, Academic Press, New York (1974).

\bibitem{SATM} J. Sabatier, O. P. Agrawal and J. A. Tenreiro Machado (eds), {\it Advances in Fractional Calculus: Theoretical Developments and Applications in Physics and Engineering}, Springer, Netherlands (2007).

\bibitem{Kilbas} A. A. Kilbas, H. M. Srivastava and J. J. Trujillo, {\it Theory and Applications of Fractional Differential Equations}, Elsevier, Amsterdam (2006).

\bibitem{Diethelm} K. Diethelm, {\it The Analysis of Fractional Differential Equations: An Application-Oriented Exposition Using Differential Operators of Caputo Type},  Springer-Verlag, Berlin Heidelberg (2010).

\bibitem{Hilfer} R. Hilfer (ed), {\it Applications of Fractional Calculus in Physics}, World Scientific, Singapore (2000).

\bibitem{Magin} R. L. Magin, {\it Fractional Calculus in Bioengineering}, Begell House Publisher (2006).

\bibitem{SKM} S. G. Samko, A. A. Kilbas, O. I. Marichev, {\it Fractional Integrals and Derivatives - Theory and Applications}, Gordon and Breach, Linghorne, PA (1993).

\bibitem{Metzler} R. Metzler and J. Klafter, {\it The random walk's guide to anomalous diffusion: a fractional dynamics approach}, Phys. Rep. {\bf 339} (2000), 1-77.

\bibitem{Metzler2} R. Metzler and J. Klafter {\it The restaurant at the end of the random walk: recent developments in the description of anomalous transport by fractional dynamics}, J. Phys. A: Math. Gen. {\bf 37} (2004), R161-R208.

\bibitem{Klages} R. Klages, G. Radons, and IM. Sokolov (eds), {\it Anomalous Transport: Foundations and Applications}, Wiley-VCH, Weinheim (2007).

\bibitem{Tarasov3} V. E. Tarasov, {\it Electromagnetic Fields on Fractals}, Mod. Phys. Lett. A {\bf 21} (2006), 1587-1600.

\bibitem{Lazo} M. J. Lazo, {\it  Gauge invariant fractional electromagnetic fields}, Phys. Lett. A {\bf 375} (2011) 3541-3546.

\bibitem{BI} E. Baskin and A. Iomin, {\it Electrostatics in fractal geometry: Fractional calculus approach}, Chaos, Solitons \& Fractals {\bf 44} (2011), 335-341.

\bibitem{Calcagni1} G. Calcagni, {\it Quantum field theory, gravity and cosmology in a fractal universe}, Journal of High Energy Physics (JHEP) {\bf 2010}, No. 3 (2010), 120. 

\bibitem{Calcagni2} G. Calcagni, {\it Geometry and field theory in multi-fractional spacetime}, Journal of High Energy Physics (JHEP) {\bf 2012}, No. 1 (2012), 65.

\bibitem{Vacaru} S. I. Vacaru, {\it Fractional Dynamics from Einstein Gravity, General Solutions, and Black Holes}, Int. J. Theor. Physics {\bf 51} (2012), 1338-1359.

\bibitem{Riewe1} F. Riewe, {\it Nonconservative Lagrangian and Hamiltonian mechanics}, Phys. Rev. E {\bf 53} (1996), 1890-1899.

\bibitem{Riewe2} F. Riewe, {\it Mechanics with fractional derivatives}, Phys. Rev. E  {\bf 55} (1997), 3581-3592.

\bibitem{Bauer} P. S. Bauer, {\it Dissipative Dynamical Systems: I}, Proc. Natl. Acad. Sci. {\bf 17} (1931), 311-314.

\bibitem{Agrawal} O. P. Agrawal, {\it Formulation of Euler-Lagrange equations for fractional variational problems}, J. Math. Anal. Appl. {\bf 272} (2002), 368-379.

\bibitem{BA} D. Baleanu and O. P. Agrawal, {\it Fractional hamilton formalism within caputo?s derivative}, Czechoslovak J. of Phys. {\bf 56}, No. 10-11 (2006), 1087-1092.

\bibitem{Cresson} J. Cresson, {\it Fractional embedding of differential operators and Lagrangian systems}, J. of Math. Phys. {\bf 48} (2007), 033504.

\bibitem{APT} R. Almeida, S. Pooseh and D. F. M. Torres, {\it Fractional variational problems depending on indefinite integrals}, Nonlinear Anal. {\bf 75} (2012), 1009-1025.

\bibitem{AT} R. Almeida and D. F. M. Torres, {\it Necessary and sufficient conditions for the fractional calculus of variations with Caputo derivatives}, Commun. Nonlinear Sci. Numer. Simul. {\bf 16} (2011), 1490-1500.

\bibitem{OMT} T. Odzijewicz, A. B. Malinowska and D. F. M. Torres, {\it Fractional variational calculus with classical and combined Caputo derivatives}, Nonlinear Anal. {\bf 75} (2012), 1507-1515.

\bibitem{MalinowskaTorres} A.B. Malinowska and D. F. M. Torres, {\it Introduction to the fractional calculus of variations}, Imperial College Press, London \& World Scientific Publishing, Singapore (2012).

\bibitem{LazoCesar} M. J. Lazo and C. E. Krumreich, {\it Lagrangian formulation for open and dissipative systems with higher-order derivatives}, submitted to Physica Scripta.

\bibitem{jerk}  S. H. Schot, {\it Jerk: The time rate of change of acceleration}, Am. J. of Phys. {\bf 46} (1978) 1090-1094.

\bibitem{jerk1} J. C. Sprott, {\it Elegant Chaos},  World Scientific Publishing Company (2010).

\bibitem{jerk2a} H. P. W. Gottlieb, {\it What is the simplest jerk function that gives chaos?}, Am. J. Phys. {\bf 64} (1996) 525.

\bibitem{jerk2b} J. C. Sprott, {\it Some simple chaotic jerk functions}, Am. J. Phys. {\bf 65} (1997) 537-543.

\bibitem{jerk2c} S. J. Linz, {\it Nonlinear dynamical models and jerky motion}, Am. J. Phys. {\bf 65} (1997) 523-525.

\bibitem{jerk2d} J. C. Sprott, {\it Simplest dissipative chaotic flow}, Phys. Lett. A {\bf 228} (1997) 271-274.

\bibitem{jerk2e} S. J. Linz and J. C. Sprott, {\it Elementary chaotic flow}, Phys. Lett. A {\bf 259} (1999) 240-245.

\bibitem{jerk2f} J. C. Sprott, {\it A new class of chaotic circuit}, Phys. Lett. A {\bf 266} (2000) 19-23.

\bibitem{jerk2g} V Patidar and K. K. Sud, {\it Identical synchronization in chaotic jerk dynamical systems}, Electronic J. of Theor. Phys. {\bf 3} No. 11 (2006) 33-70.

\bibitem{TRG} K. M. Tarawneh, E. M. Rabei, H. B. Ghassib, {\it Lagrangian and Hamiltonian Formulations of the Damped Harmonic Oscillator using Caputo Fractional Derivative}, J. Dyn. Syst. Geom. Theories {\bf 8} (2010) 59-70.



\end{thebibliography}
\end{document}